\documentclass[12pt,A4paper]{article}
\usepackage{amssymb,amsfonts,amsthm,amsmath}
\usepackage{graphicx}
\usepackage{placeins}
\usepackage{amssymb}

\usepackage{epsfig}
\usepackage{placeins}

\newtheorem{theorem}{Theorem}
\newtheorem{lemma}{Lemma}

\begin{document}

\baselineskip=24pt

\begin{center} {\Large{\textbf {Variable-width confidence intervals in Gaussian regression and penalized maximum
likelihood estimators}}}
\end{center}

\medskip

\smallskip


\begin{center}
\large{Davide Farchione and Paul Kabaila$^*$}
\end{center}

\begin{center}
{\it Department of Mathematics and Statistics, La Trobe University,  Australia}
\end{center}

\vspace{13cm}

\noindent $^*$ Author to whom correspondence should be addressed.
Department of Mathematics and Statistics, La Trobe University,
Victoria 3086, Australia.
Tel.: +61 3 9479 2594, Fax:
 +61 3 9479 2466,
 {E-mail:} P.Kabaila@latrobe.edu.au

\newpage

\begin{center}
{\large ABSTRACT}
\end{center}

\smallskip

\noindent Hard thresholding, LASSO , adaptive LASSO and SCAD
point estimators have been suggested for use in the linear regression context when
most of the components of the regression parameter vector are believed to be zero,
a sparsity type of assumption.
P\"otscher and Schneider, 2010, {\sl Electronic Journal of Statistics},
have considered the properties of
fixed-width confidence intervals that include one of these point estimators
(for all possible data values).
They consider a normal linear regression model with orthogonal regressors and show that these confidence
intervals are longer than the standard confidence interval
(based on the maximum likelihood estimator) when the tuning parameter
for these point estimators is chosen to lead to either conservative
or consistent model selection.
We extend this analysis to the case of {\sl variable-width} confidence intervals
that include one of these point estimators
(for all possible data values). In consonance with these findings of P\"otscher and Schneider,
we find that these confidence intervals
perform poorly by comparison with the standard confidence interval,
when the tuning parameter for these point estimators
is chosen to lead to consistent model selection.
However, when the tuning parameter for these point estimators is chosen to lead to conservative
model selection, our conclusions differ from those of P\"otscher and Schneider.
 We consider the variable-width
confidence intervals of Farchione and Kabaila, 2008, {\sl Statistics
\& Probability Letters}, which have advantages over the standard confidence
interval in the context that there is a belief in a sparsity type of assumption.
These variable-width confidence intervals
are shown to include the hard thresholding, LASSO, adaptive LASSO
and SCAD estimators (for all possible data values) provided that
the tuning parameters for these estimators are chosen to
belong to an appropriate interval.

\newpage

\section{Introduction}

Hard-thresholding, LASSO (Tibshirani \cite{Tibshirani}), adaptive LASSO (Zou \cite{Zou}) and SCAD
(Fan and Li \cite{Fan_Li})
point estimators have been suggested for use in the linear regression context when
most of the components of the regression parameter vector are believed to be zero,
a sparsity type of assumption.
P\"otscher and Schneider \cite{Potscher_Schneider} ask to what extent these
point estimators can be used as the basis for confidence intervals for these
components. They consider the properties of
fixed-width confidence intervals that are constrained to include one of these
point estimators (for all possible data values).
They do this in the context of a normal linear regression model with orthogonal regressors
for both the case that (a) the error variance is assumed known and (b)
the error variance is estimated by the usual unbiased estimator
obtained by fitting the full model to the data.
P\"otscher and Schneider \cite{Potscher_Schneider} show that these confidence
intervals are longer than the standard confidence interval
based on the maximum likelihood estimator, when the tuning parameter
for these point estimators is chosen to lead to either conservative
or consistent model selection. By consistent model selection, we mean
that the selected model is the true model with probability approaching
1 as $n \rightarrow \infty$, where $n$ denotes the dimension of the
response vector. By conservative model selection, we mean
a model selection that (a) is not consistent and (b) is such
that the selected model includes the true model with probability approaching
1 as $n \rightarrow \infty$.

To what extent are these findings due to the requirement that these confidence
intervals have fixed widths? A variable-width confidence interval
based on a given point estimator has the property that this
confidence interval includes this point estimator, for all possible
data values. We first consider the case that the tuning parameter
for these point estimators is chosen to lead to consistent model selection.
In Section 3, we present a new result that shows that variable-width confidence
intervals that include one of these point estimators (for all possible data
values) must perform poorly by comparison with the standard confidence interval.
In this case, our conclusions are similar to those in \cite{Potscher_Schneider}.
This is perhaps not surprising, given the results of Kabaila \cite{Kabaila}
and P\"otscher \cite{Potscher}.

Next, we consider the case that the tuning parameter for these point estimators
is chosen to lead to conservative model selection.
P\"otscher and Schneider \cite{Potscher_Schneider} find that
fixed-width confidence intervals that are constrained to include one of these
point estimators (for all possible data values) are longer than the standard
confidence interval.
This may be interpreted as a negative finding for these point estimators.
Yet, these point estimators have some very attractive features.
Figure 9 of \cite{Tibshirani} shows contours of constant value of $|\beta_1|^q + |\beta_2|^q$
for $q=4, 2, 1, 0.5$ and 0.1. As
Tibshirani \cite{Tibshirani} states, ``The lasso corresponds to $q=1$.''
and ``The value $q=1$ has the advantage of being closer to subset selection
than is ridge regression ($q=2$) and is also the smallest value of $q$
giving a convex region.''. The LASSO estimator has the attractive feature that
it is a continuous function of the data. Like the LASSO, the
adaptive LASSO and the SCAD estimators use a thresholding rule
that sets estimated coefficients with small magnitudes to zero. The adaptive LASSO and
the SCAD estimators also have the attractive features that (a) they are
continuous functions of the data and (b) they are nearly unbiased when
the true unknown parameter has large magnitude (\cite{Fan_Li}, \cite{Zou}).
How do we resolve the apparent conflict between the findings of
\cite{Potscher_Schneider} and the existence of these
very attractive features? We show that this finding can be explained (at least in part)
by the requirement in \cite{Potscher_Schneider} that the confidence intervals have
fixed widths.

Following \cite{Potscher_Schneider}, we consider a normal linear regression
model with orthogonal regressors for both the case that (a) the error variance is assumed known and (b)
the error variance is estimated by the usual unbiased estimator
obtained by fitting the full model to the data.
It is plausible that the case that the error variance is known
amounts essentially to the assumption that the error variance is estimated with
great accuracy. In Appendix B, we provide a precise motivation for considering
the known error variance case. In Section 4, we
consider the variable-width
confidence intervals of Farchione and Kabaila \cite{Farchione_Kabaila},
in the known error variance case.
These confidence intervals are shown to
have advantages over the standard confidence
interval when there is a belief in a sparsity type of assumption.
These variable-width confidence intervals
are shown to include the hard-thresholding, LASSO, adaptive LASSO
and SCAD estimators (for all possible data values) provided that
the tuning parameters for these estimators are chosen to
belong to an appropriate interval. In Section 5, we consider
the extension of these results
to the case that the error variance is estimated by the usual unbiased estimator
obtained by fitting the full model to the data.

\section{The model and the point estimators considered}

We consider a normal linear regression model with orthogonal regressors.
As pointed out in \cite{Potscher_Schneider}, without loss of generality we may
suppose that the data $Y_1, \ldots, Y_n$ are independent and identically $N(\theta, \sigma^2)$
distributed,
where $\theta \in \mathbb{R}$ and $\sigma > 0$. We use lower case to denote the observed
value of a random variable.
We also use a similar notation to that used in \cite{Potscher_Schneider} for the
hard thresholding, LASSO and adaptive LASSO estimators. Namely, the hard thresholding
estimator $\tilde{\Theta}_H$ is given by
\begin{equation*}
\tilde{\Theta}_H = \bar{Y} \, \textbf{1}(|\bar{Y}| > \hat{\Sigma} \eta_n) =
\begin{cases}
0 &\text{if } |\bar{Y}| \le \hat{\Sigma} \eta_n \\
\bar{Y}       &\text{if } |\bar{Y}| > \hat{\Sigma} \eta_n
\end{cases}
\end{equation*}
where the tuning parameter $\eta_n$ is a positive real number, $\bar{Y}= n^{-1} \sum_{i=1}^n Y_i$
and $\hat{\Sigma}^2 = (n-1)^{-1}\sum_{i=1}^n (Y_i - \bar{Y})^2$. The LASSO estimator
$\tilde{\Theta}_S$ is given by
\begin{equation*}
\tilde{\Theta}_S = \text{sign}(\bar{Y}) \, \big(|\bar{Y}| > \hat{\Sigma} \eta_n \big)_+ =
\begin{cases}
- \max\{|\bar{Y}| - \hat{\Sigma} \eta_n, 0\} &\text{if } \bar{Y}<0 \\
0 &\text{if } \bar{Y}=0 \\
\max\{|\bar{Y}| - \hat{\Sigma} \eta_n, 0\} &\text{if } \bar{Y}>0
\end{cases}
\end{equation*}
where sign$(x)$ is equal to $-1$ for $x<0$, 0 for $x=0$ and 1 for $x>0$ and
$x_+ = \max\{x,0\}$. The adaptive LASSO estimator $\tilde{\Theta}_A$ is given by
\begin{equation*}
\tilde{\Theta}_A = \bar{Y} \, \big(1 - \hat{\Sigma}^2 \eta_n^2/\bar{Y}^2 \big)_+ =
\begin{cases}
0 &\text{if } |\bar{Y}| \le \hat{\Sigma} \eta_n \\
\bar{Y} - \displaystyle{\frac{\hat{\Sigma}^2 \eta_n^2}{\bar{Y}}} &\text{if } |\bar{Y}| > \hat{\Sigma} \eta_n
\end{cases}
\end{equation*}
We also consider the following SCAD estimator $\tilde{\Theta}_C$
\begin{equation*}
\tilde{\Theta}_C =
\begin{cases}
\text{sign}(\bar{Y}) \, \big(|\bar{Y}| - \hat{\Sigma} \eta_n \big)_+
&\text{if } |\bar{Y}| \le 2 \hat{\Sigma} \eta_n \\
((a-1) \bar{Y} - \text{sign}(\bar{Y}) a \hat{\Sigma} \eta_n)/(a-2)
&\text{if }  2 \hat{\Sigma} \eta_n < |\bar{Y}| \le a \hat{\Sigma} \eta_n \\
\bar{Y}
&\text{if }  |\bar{Y}| > a \hat{\Sigma} \eta_n
\end{cases}
\end{equation*}
where $a=3.7$ (see p.1351 of \cite{Fan_Li} for a motivation for this choice of $a$).

\section{Variable-width confidence intervals based on the point estimators
when the tuning parameter is chosen for consistent model selection}

In this section, we suppose that $\eta_n \rightarrow 0$ and $\sqrt{n} \, \eta_n \rightarrow \infty$,
as $n \rightarrow \infty$.
In other words, we suppose that the tuning parameter $\eta_n$ is chosen so as to lead to consistent
model selection. In this case, for example, the probability that $\tilde{\Theta}_H$ is equal to 0
approaches 1 for $\theta=0$, whilst $\tilde{\Theta}_H$ converges in probability to
$\theta$ for $\theta \ne 0$ (as $n \rightarrow \infty$).
For clarity, in this section we will use the subscript $n$ to make explicit a
dependence on $n$. Let $\tilde{\theta}_n(\bar{y}_n, \hat{\sigma}_n)$ denote a point estimate of
$\theta$ that satisfies the condition that if $|\bar{y}_n| \le \hat{\sigma}_n \eta_n$ then
$\tilde{\theta}_n(\bar{y}_n, \hat{\sigma}_n)=0$. The estimates $\tilde{\theta}_H$,
$\tilde{\theta}_S$ and $\tilde{\theta}_A$ satisfy this condition.
With a small change of notation,
the estimate $\tilde{\theta}_C$ also satisfies this condition.
The standard $1-\alpha$ confidence interval for $\theta$ is
\begin{equation*}
J_n = \Big [ \bar{Y}_n - t(n-1) \hat{\Sigma}_n / \sqrt{n}, \, \bar{Y}_n + t(n-1) \hat{\Sigma}_n / \sqrt{n} \Big ]
\end{equation*}
where the quantile $t(m)$ is defined by the requirement that $P \big(-t(m) \le T \le t(m) \big) = 1 - \alpha$
for $T \sim t_m$.

A variable-width confidence interval
based on the point estimate $\tilde{\theta}_n(\bar{y}_n, \hat{\sigma}_n)$
has the property that this
confidence interval includes this point estimate, for all
possible data values.
Consider the confidence interval
\begin{equation*}
D_n(\bar{Y}_n, \hat{\Sigma}_n)
= \big [ \ell_n(\bar{Y}_n, \hat{\Sigma}_n), \, u_n(\bar{Y}_n, \hat{\Sigma}_n) \big]
\end{equation*}
for $\theta$, that is required to satisfy the following conditions for all $n$:

\begin{enumerate}

\item[(a)]  $\tilde{\theta}_n(\bar{y}_n, \hat{\sigma}_n) \in D_n(\bar{y}_n, \hat{\sigma}_n)$
for all $(\bar{y}_n, \hat{\sigma}_n) \in \mathbb{R} \times (0, \infty)$. In other words,
the confidence interval $D_n$ contains the estimate $\tilde{\theta}_n$, for all possible
data values.

\item[(b)]  $P_{\theta, \sigma} \big ( \theta \in D_n(\bar{Y}_n, \hat{\Sigma}_n) \big)
\ge 1 - \alpha$ for all $(\theta, \sigma) \in \mathbb{R} \times (0, \infty)$. In other words,
$D_n$ is a $1-\alpha$ confidence interval for $\theta$.

\end{enumerate}

\noindent The following result shows that this confidence
interval performs very poorly by comparison with $J_n$,
the standard $1-\alpha$ confidence interval for $\theta$.

\begin{theorem}

Let $\theta_n = \sigma \, \eta_n/2$. For each $\sigma \in (0, \infty)$,
\begin{equation*}
\frac{E_{\theta_n,\sigma} \big(\text{length of } D_n(\bar{Y}_n, \hat{\Sigma}_n) \big)}
{E_{\theta_n,\sigma} \big(\text{length of standard } 1-\alpha \text{ confidence interval } J_n \big)}
\rightarrow \infty
\end{equation*}
as $n \rightarrow \infty$.

\end{theorem}
\noindent The proof of this theorem is presented in Appendix A.

\section{Variable-width confidence intervals of Farchione and Kabaila when
the error variance is known}

Consider the ``known error variance case''. The motivation for considering
this case is given in Appendix B. Suppose that $\sigma^2$ is known. Consider the $1-\alpha$ confidence interval
for $\theta$, put forward by Farchione and Kabaila \cite{Farchione_Kabaila},
that has the form
\begin{equation}
\label{CI_C}
C = \left [ -\frac{\sigma}{\sqrt{n}} b \left ( -\frac{\bar{Y}}{\sigma/\sqrt{n}} \right ), \,
\frac{\sigma}{\sqrt{n}} b \left ( \frac{\bar{Y}}{\sigma/\sqrt{n}} \right ) \right ]
\end{equation}
where the function $b$ satisfies $b(x) \ge -b(-x)$ for all $x \in \mathbb{R}$.
This constraint is required to ensure that the upper endpoint of this confidence interval is never less
than the lower endpoint. This particular form of confidence interval is motivated by the invariance arguments
presented in Section 4 of \cite{Farchione_Kabaila}. The standard $1-\alpha$ confidence interval for
$\theta$ is $I = \big[\bar{Y} - z \, \sigma / \sqrt{n}, \, \bar{Y} + z \, \sigma / \sqrt{n} \big]$,
where the quantile $z$ is defined by the requirement that $P(-z \le Z \le z) = 1 - \alpha$
for $Z \sim N(0,1)$. Note that this confidence interval can be expressed in the form $C$.

The coverage probability and expected length properties of the confidence interval $C$ are conveniently
examined by applying the same change of scale (by multiplying by $\sqrt{n}/\sigma$) to the parameter
$\theta$, the estimator $\bar{Y}$, the confidence interval $C$ and the standard confidence interval $I$.
Define $\psi = (\sqrt{n}/\sigma) \theta$, $X = (\sqrt{n}/\sigma) \bar{Y}$,
\begin{equation}
\label{CI_Cstar}
C^* = \frac{\sqrt{n}}{\sigma} C = [-b(-X), b(X)],
\end{equation}
and $I^* = (\sqrt{n}/\sigma) I = [X-z, X+z]$.
Note that $X \sim N(\psi, 1)$.
We consider $C^*$ to be a confidence interval for
$\psi$, based on $X$. The standard $1-\alpha$ confidence interval for $\psi$ (based on $X$) is
$I^*$. Note that $P_{\theta,\sigma} (\theta \in C) = P_{\psi}(\psi \in C^*)$
and
\begin{equation*}
\frac{E_{\theta, \sigma} (\text{length of } C)}{\text{length of } I}
= \frac{E_{\psi} (\text{length of } C^*)}{\text{length of } I^*},
\end{equation*}
for $\psi= (\sqrt{n}/\sigma) \theta$.

Following \cite{Farchione_Kabaila}, we
assess $C^*$, for parameter value $\psi$, using the relative efficiency
\begin{equation*}
e(\psi) = \left ( \frac{E_{\psi} (\text{length of } C^*)}{\text{length of } I^*} \right )^2
= \left (  \frac{E_{\psi}(\text{length of } C^*)}{2 z} \right )^2.
\end{equation*}
This is a measure of the efficiency of the standard $1-\alpha$ confidence interval $I^*$
by comparison with the efficiency of the $1-\alpha$ confidence interval $C^*$.
The relative efficiency $e(\psi)$ is the ratio (sample size used for $C^*$)/(sample size
used for $I^*$) such that $E_{\psi} (\text{length of } C^*) = \text{length of } I^*$
(cf p.555 of \cite{Pratt}).
Farchione and Kabaila \cite{Farchione_Kabaila} use the methodology of Pratt \cite{Pratt},
with a new weight function determined by a parameter $w$,
to find a confidence interval $C^*$ such that $e(0)$ is minimized, while ensuring that
$\max_{\psi} e(\psi)$ is not too large. In other words, if $\psi$ happens to be 0 then
$C^*$ performs better than the standard $1-\alpha$ confidence interval $I^*$. On the other hand,
if $\psi \ne 0$ then the worst possible performance of $C^*$ is $\max_{\psi} e(\psi)$,
which is not too large.
In addition, this confidence interval has endpoints that approach
the endpoints of the standard $1-\alpha$ confidence interval $I^*$ as $|x| \rightarrow \infty$.
This implies that $e(\psi) \rightarrow 1$ as $|\psi| \rightarrow \infty$.
We have chosen $w = 0.1$ and $1-\alpha=0.95$. The coverage probability $P_{\psi}(\psi \in C^*)$ is
0.95 for all $\psi$. The relative efficiency $e(\psi)$ of $C^*$ for this case is shown in Figure 1.
For comparison, the 0.95 confidence interval described on p.555 of \cite{Pratt} has
relative efficiency 0.72 at $\psi=0$. This, however, comes at the very high cost of the relative
efficiency diverging to $\infty$ as $|\psi| \rightarrow \infty$.

\FloatBarrier

\begin{figure}[h]\hspace{30mm}
    \scalebox{0.35}{\includegraphics[]{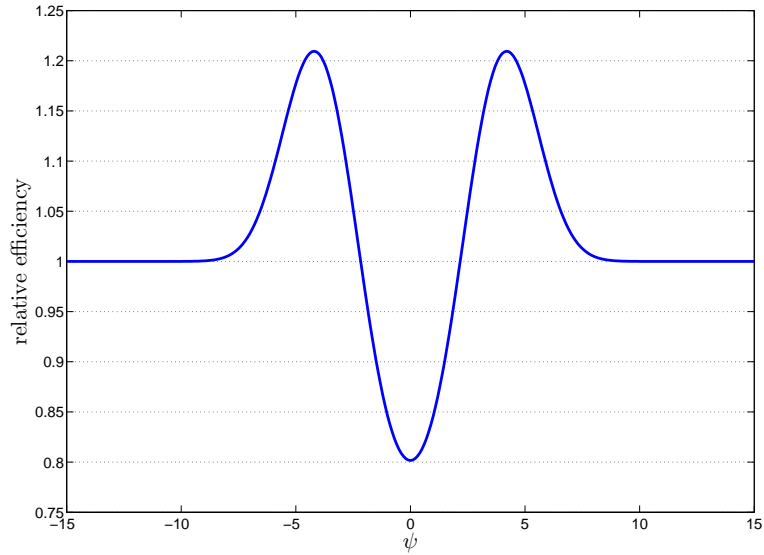}}
    \caption{\textit{\small{Plot of the efficiency of the standard 95\% confidence interval
    by comparison with the Farchione and Kabaila 95\% confidence
    interval (for $w = 0.1$) as a function of $\psi$.}}}\label{}
\end{figure}

\FloatBarrier

We now consider the properties of the confidence interval $C^*$ in the context
that most of the components of the regression parameter vector are believed to be zero,
a sparsity type of assumption. Firstly, suppose that a large majority of the components of the
regression parameter vector are zero.
In this case, $C^*$ compares very favourably with the standard $1-\alpha$ confidence interval.
If $\psi=0$, corresponding
to one of the large majority of the components of the regression parameter vector that are zero, then
$e(\psi)$ is approximately 0.8. On the other hand, if $\psi \ne 0$,
corresponding to one of the small minority of components of the regression parameter
vector that are non-zero, then the maximum possible value of $e(\psi)$ is approximately 1.2.
Secondly, in the ``best of all possible worlds''
scenario that a large majority of the components of the regression parameter vector are zero and
the remaining components have large magnitudes, $C^*$ may be said to effectively
dominate the standard $1-\alpha$ confidence interval. If $\psi=0$, corresponding
to one of the large majority of the components of the regression parameter vector that is zero, then
$e(\psi)$ is approximately equal to 0.8. On the other hand, if $|\psi|$ is large,
corresponding to one of the small minority of the components of the regression parameter vector that has large
magnitude, then $e(\psi)$ is approximately equal to 1.
We conclude that $C^*$ has advantages over the standard $1-\alpha$ confidence
interval $I^*$ when a sparsity type of assumption holds.

\section{Variable-width confidence intervals based on the point estimators when
the tuning parameter is chosen for conservative model selection and the error variance is known}

In this section, we suppose that $\eta_n \rightarrow 0$.
We also suppose that there exists a positive integer $N$ and $a_{\ell}$ and $a_u$
(satisfying $0 < a_{\ell} < a_u < \infty$), such that
$\sqrt{n} \, \eta_n \in [a_{\ell}, a_u]$ for all $n > N$. This includes the particular case that
$\sqrt{n} \, \eta_n \rightarrow a$ ($0< a < \infty$),
as $n \rightarrow \infty$.
In other words, we suppose that the tuning parameter $\eta_n$ is chosen so as to lead to conservative
model selection.
We consider the ``known error variance case''. The motivation for considering
this case is given in Appendix B.

Suppose that $\sigma^2$ is known. We consider the conditions under which the point estimate
\begin{equation*}
\hat{\theta}_H =
\begin{cases}
0 &\text{if} \ \ |\bar{y}| \le \sigma \, \eta_n \\
\bar{y} &\text{if} \ \ |\bar{y}| > \sigma \, \eta_n
\end{cases}
\end{equation*}
of $\theta$ belongs in the confidence interval
$C$ (defined by \eqref{CI_C})
for all $(\bar{y}, \sigma) \in \mathbb{R} \times (0, \infty)$.
Define $\tau_n = \sqrt{n} \, \eta_n$.
As in Section 4, multiply the estimate $\hat{\theta}_H$ and the confidence interval $C$ by $\sqrt{n}/\sigma$,
to obtain
\begin{equation*}
\hat{\psi}_H = \frac{\sqrt{n}}{\sigma} \hat{\theta}_H =
\begin{cases}
0 &\text{if} \ \ |x| \le \tau_n \\
x  &\text{if} \ \ |x| > \tau_n
\end{cases}
\end{equation*}
and $C^* = (\sqrt{n}/\sigma) C$ (see \eqref{CI_Cstar}).
Obviously, $\hat{\theta}_H \in C$ for all $(\bar{y}, \sigma) \in \mathbb{R} \times (0,\infty)$
is equivalent to $\hat{\psi}_H \in C^*$ for all $x \in \mathbb{R}$.
There exists a positive number $c_H$ such that, for every $\tau_n \in (0, c_H]$, the following is true:
$\hat{\psi}_H \in C^*$ for all $x\in \mathbb{R}$. Similar statements hold for the other point estimates $\hat{\theta}_S$,
$\hat{\theta}_A$ and $\hat{\theta}_C$ (the corresponding estimators are defined towards the end of Appendix B).

Define $\hat{\psi}_S = (\sqrt{n}/\sigma) \hat{\theta}_S$,
$\hat{\psi}_A = (\sqrt{n}/\sigma) \hat{\theta}_A$ and
$\hat{\psi}_C = (\sqrt{n}/\sigma) \hat{\theta}_C$.
We have computed the maximum values of $\tau_n$ such that $\hat{\psi}_H$,  $\hat{\psi}_S$,
$\hat{\psi}_A$ and
$\hat{\psi}_C$ are in the interval $C^*$ (for all $x$). In each case this maximum value was found
to be 1.96. Figures 2 and 3 show the values of the estimator as a function of $x$ for this maximum value,
together with the endpoints of the confidence interval $C^*$ as functions of $x$.

\bigskip

\FloatBarrier

\begin{figure}[h]\hspace{22mm}
\scalebox{0.48}
    {\includegraphics[]{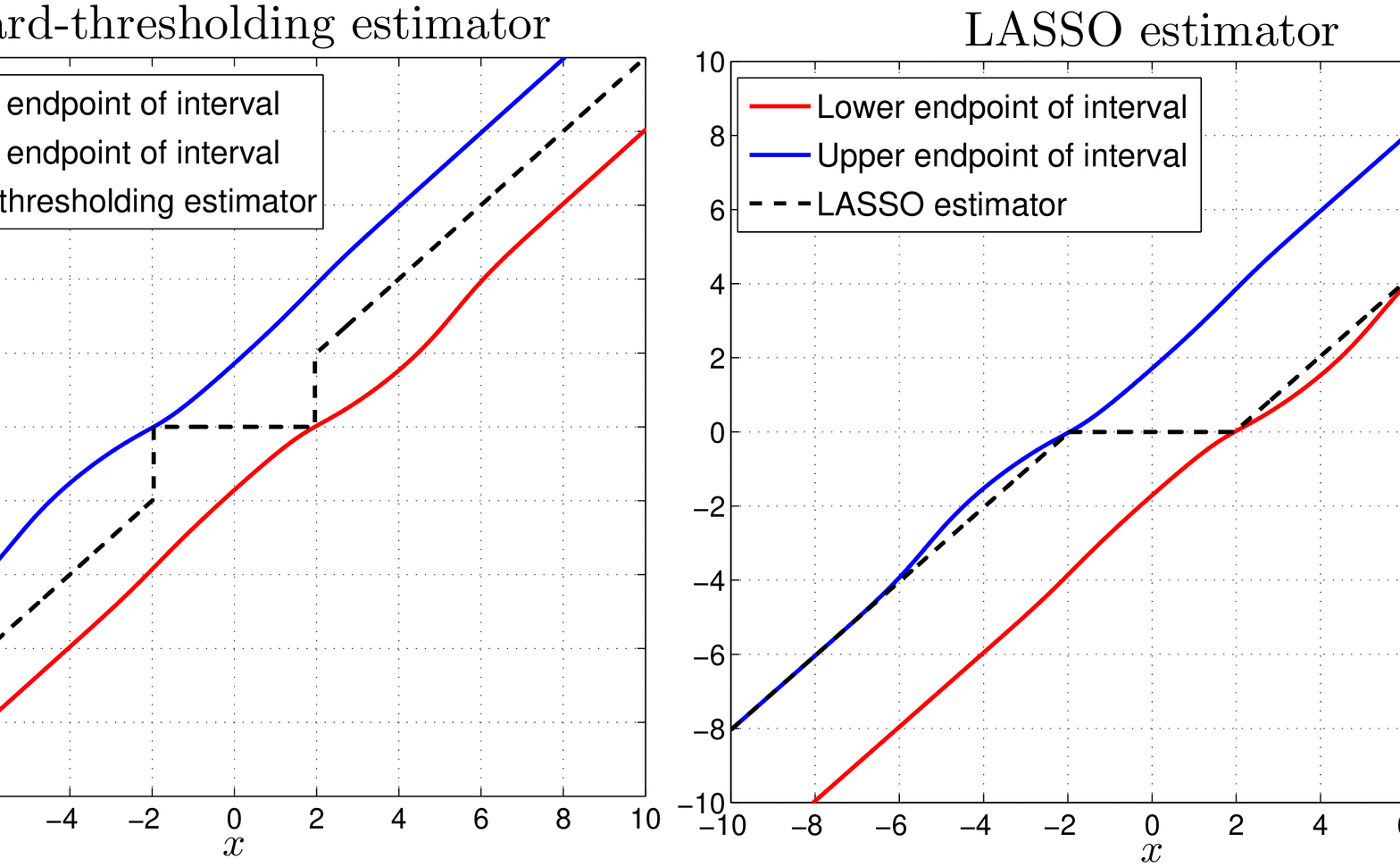}}
    \caption{\textit{\small{The left and right panels show the hard-thresholding estimate $\hat{\psi}_H$
    and the LASSO estimate $\hat{\psi}_S$, respectively, as functions of $x$ (for $\tau_{n} = 1.96$).
      Also shown, in both panels, is the
    Farchione and Kabaila 95\% confidence interval $C^*$ as a function of $x$ (for $w = 0.1$).}}}\label{}
\end{figure}

\FloatBarrier
\begin{figure}[h]\hspace{22mm}
    \scalebox{0.48}
    {\includegraphics[]{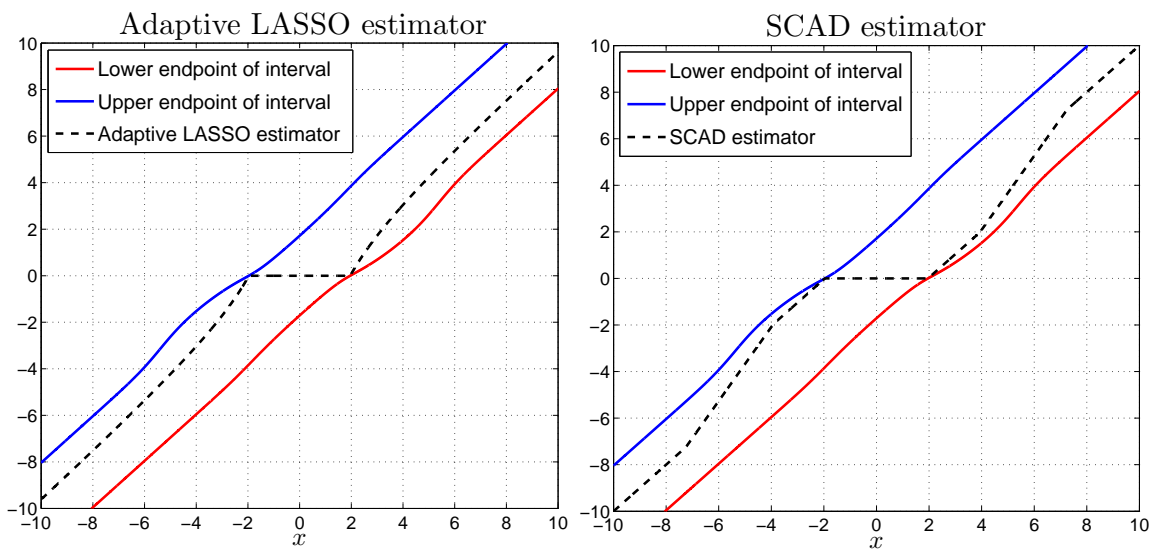}}
    \caption{\textit{\small{The left and right panels show the Adaptive LASSO estimate $\hat{\psi}_A$ (for $\tau_{n} = 1.96$)
    and the SCAD estimate $\hat{\psi}_S$ (for $\tau_{n} = 1.96$, $a = 3.7$), respectively, as functions of $x$ .
      Also shown, in both panels, is the
    Farchione and Kabaila 95\% confidence interval $C^*$ as a function of $x$ (for $w = 0.1$).
    .}}}\label{}
\end{figure}

\FloatBarrier

\section{Variable-width confidence intervals of Farchione and Kabaila and the point estimators when
the tuning parameter is chosen for conservative model selection and the error variance is unknown}

Suppose that the error variance $\sigma^2$ is unknown.
Consider the $1-\alpha$ confidence interval for $\theta$, put forward in Section 5 of \cite{Farchione_Kabaila},
that has the form
\begin{equation*}
D = \left [ -\frac{\hat{\Sigma}}{\sqrt{n}} \, b \left ( -\frac{\bar{Y}}{\hat{\Sigma}/\sqrt{n}} \right ), \,
\frac{\hat{\Sigma}}{\sqrt{n}} \, b \left ( \frac{\bar{Y}}{\hat{\Sigma}/\sqrt{n}} \right ) \right ]
\end{equation*}
where the function $b$ satisfies $b(x) \ge -b(-x)$ for all $x \in \mathbb{R}$.
This constraint is required to ensure that the upper endpoint of this confidence interval is never less
than the lower endpoint. This particular form of confidence interval can be motivated by invariance arguments
similar to those presented in Section 4 of \cite{Farchione_Kabaila}. The standard $1-\alpha$ confidence interval for
$\theta$ is
\begin{equation*}
J = \big [ \bar{Y} - t(n-1) \hat{\Sigma} / \sqrt{n}, \, \bar{Y} + t(n-1) \hat{\Sigma} / \sqrt{n} \big ]
\end{equation*}
where the quantile $t(m)$ is defined by the requirement that $P(-t(m) \le T \le t(m)) = 1 - \alpha$
for $T \sim t_m$. Note that this confidence interval can be expressed in the form $D$.

Define $R= \hat{\Sigma}/\sigma$.
The coverage probability and expected length properties of the confidence interval $D$ are conveniently
examined by applying the same change of scale (by multiplying by $\sqrt{n}/\sigma$) to the parameter
$\theta$, the estimator $\bar{Y}$, the confidence interval $D$ and the standard confidence interval $J$.
Define $\psi = (\sqrt{n}/\sigma) \theta$, $X = (\sqrt{n}/\sigma) \bar{Y}$,
\begin{equation*}
D^* = \frac{\sqrt{n}}{\sigma} D = \big [ -R \, b(-X/R), \, R \, b(X/R) \big ].
\end{equation*}
and $J^* = (\sqrt{n}/\sigma) J = [X - t(n-1) R, X + t(n-1) R]$.
Note that $X$ and $R$ are independent random variables and that
$X \sim N(\psi, 1)$.
As noted in Appendix B, the coverage probability and expected length properties
of $D$ are conveniently evaluated using the fact that
\begin{equation*}
P_{\theta, \sigma} (\theta \in D) = P_{\psi} (\psi \in D^*),
\end{equation*}
and
\begin{equation*}
\frac{E_{\theta, \sigma} (\text{length of } D)}{E_{\theta, \sigma} (\text{length of } J)}
= \frac{E_{\psi} ( \text{length of } D^*)}{E_{\theta, \sigma} (\text{length of } J^*)}
\end{equation*}
for $\psi= (\sqrt{n}/\sigma) \theta$.

Following \cite{Farchione_Kabaila}, we
assess $D^*$, for parameter value $\psi$, using the relative efficiency
\begin{equation*}
e(\psi) = \left ( \frac{E_{\psi} ( \text{length of } D^*)}{E_{\theta, \sigma} (\text{length of } J^*)} \right )^2
= \left (  \frac{E_{\psi}(\text{length of } D^*)}{2 t(n) E(R)} \right )^2.
\end{equation*}
This is a measure of the efficiency of the standard $1-\alpha$ confidence interval $J^*$
by comparison with the efficiency of the $1-\alpha$ confidence interval $D^*$.
Farchione and Kabaila \cite{Farchione_Kabaila} present (in Section 6) a computational methodology
with a weight function determined by a parameter $w$,
to find a confidence interval $D^*$ such that $e(0)$ is minimized, while ensuring that
$\max_{\psi} e(\psi)$ is not too large. In other words, if $\psi$ happens to be 0 then
$D^*$ performs better than the standard $1-\alpha$ confidence interval $J^*$. On the other hand,
if $\psi \ne 0$ then the worst possible performance of $D^*$ is $\max_{\psi} e(\psi)$,
which is not too large.
In addition, the confidence interval $D^*$ has endpoints that are the same as
the endpoints of the standard $1-\alpha$ confidence interval $J^*$ for sufficiently large $|X|/R$.
This implies that $e(\psi) \rightarrow 1$ as $|\psi| \rightarrow \infty$.
Farchione and Kabaila \cite{Farchione_Kabaila} found computationally that for the same choice of
parameter $w$, the confidence intervals $C^*$ and $D^*$ have similar
relative efficiencies (as function of $\psi$), provided that $n$ is not small. This
is illustrated by Figure 2 of \cite{Farchione_Kabaila}.
Theoretical support for this computational finding is provided by Theorem 2 of Appendix B
of the present paper.

As in Section 5, suppose that the tuning parameter $\eta_n$ is chosen so as to lead to conservative
model selection. We consider the conditions under which the point estimate
\begin{equation*}
\tilde{\theta}_H =
\begin{cases}
0 &\text{if} \ \ |\bar{y}| \le \hat{\sigma} \, \eta_n \\
\bar{y} &\text{if} \ \ |\bar{y}| > \hat{\sigma} \, \eta_n
\end{cases}
\end{equation*}
of $\theta$ belongs in the confidence interval $D$ (observed value)
for all $(\bar{y}, \hat{\sigma}) \in \mathbb{R} \times (0, \infty)$.
Define $\tau_n = \sqrt{n} \, \eta_n$.
Multiply the estimate $\tilde{\theta}_H$ and the confidence interval $D$ by $\sqrt{n}/\hat{\sigma}$,
to obtain
\begin{align*}
\tilde{\psi}_H &= \frac{\sqrt{n}}{\hat{\sigma}} \tilde{\theta}_H =
\begin{cases}
0 &\text{if} \ \ |\tilde{x}| \le \tau_n \\
\tilde{x}  &\text{if} \ \ |\tilde{x}| > \tau_n
\end{cases} \\
\\
\tilde{D} &= \frac{\sqrt{n}}{\hat{\sigma}} D = [-b(-\tilde{x}), b(\tilde{x})],
\end{align*}
where $\tilde{x} = (\sqrt{n}/\hat{\sigma}) \bar{y}$.
Obviously, $\tilde{\theta}_H \in D$ for all $(\bar{y}, \hat{\sigma}) \in \mathbb{R} \times (0, \infty)$ is equivalent to
$\tilde{\psi}_H \in \tilde{D}$ for all $\tilde{x}\in \mathbb{R}$.
There exists a positive number $\tilde{c}_H$ such that, for every $\tau_n \in (0, \tilde{c}_H]$, the following is true:
$\tilde{\psi}_H \in \tilde{D}$ for all $\tilde{x}$. Similar statements hold for the other point estimates $\tilde{\theta}_S$,
$\tilde{\theta}_A$ and $\tilde{\theta}_C$. As note earlier, the computational results of \cite{Farchione_Kabaila}
and Theorem 2 of Appendix B, suggest that (provided that $n$ is not small)
the situation here is very similar to that described in Section 5 and Figures 2 and 3.
In other words, we expect that $\tilde{c}_H \approx c_H$ (where $c_H$ is defined in Section 5),
provided that $n$ is not small.

\section{Conclusion}

The results of this paper confirm, yet again, that the hard-thresholding, LASSO, adaptive LASSO and SCAD
point estimators form a very poor foundation for confidence interval construction when the tuning parameter
for these estimators is chosen to lead to {\sl consistent} model selection.
However, the results of this paper do not, by any means, rule out the use of these point estimators as the
foundation for confidence interval construction when the tuning parameter
for these estimators is chosen to lead to {\sl conservative} model selection.

\appendix

\section{Proof of Theorem 1}

Define the event
$A_n = \big \{ |\bar{Y}_n| \le \hat{\Sigma}_n \eta_n \big \}$. By the law of total probability,
\begin{equation*}
P_{\theta, \sigma} \big ( \{ \theta \in D_n(\bar{Y}_n, \hat{\Sigma}_n) \} \cap A_n \big ) +
P_{\theta, \sigma} \big ( \{ \theta \in D_n(\bar{Y}_n, \hat{\Sigma}_n) \} \cap A_n^c \big )
\ge 1 - \alpha \ \  \text{for all} \ \ (\theta, \sigma).
\end{equation*}
In particular,
\begin{equation*}
P_{\theta_n, \sigma} \big ( \{ \theta_n \in D_n(\bar{Y}_n, \hat{\Sigma}_n) \} \cap A_n \big ) +
P_{\theta_n, \sigma} \big ( \{ \theta_n \in D_n(\bar{Y}_n, \hat{\Sigma}_n) \} \cap A_n^c \big )
\ge 1 - \alpha \ \ \text{for all} \ \ \sigma.
\end{equation*}
Define the event $B_n = \big \{ u_n(\bar{Y}_n, \hat{\Sigma}_n) \ge \theta_n \big \}$.
When the event $A_n$ occurs, $\ell_n(\bar{Y}_n, \hat{\Sigma}_n) \le 0$ and so
\begin{equation*}
P_{\theta_n, \sigma} \big ( \{ \theta_n \in D_n(\bar{Y}_n, \hat{\Sigma}_n) \} \cap A_n \big ) =
P_{\theta_n, \sigma} \big ( B_n \cap A_n \big ) \ \ \text{for all} \ \ \sigma.
\end{equation*}
Thus, for each $\sigma \in (0, \infty)$,
\begin{align*}
P_{\theta_n, \sigma} \big ( B_n \cap A_n \big ) &\ge 1 - \alpha
- P_{\theta_n, \sigma} \big ( \{ \theta_n \in D_n(\bar{Y}_n, \hat{\Sigma}_n) \} \cap A_n^c \big ) \\
&\ge 1 - \alpha
- P_{\theta_n, \sigma} \big ( A_n^c \big ).
\end{align*}

\begin{lemma}

For each $\sigma \in (0, \infty)$, $P_{\theta_n, \sigma} \big ( A_n^c \big ) \rightarrow 0$ as $n \rightarrow \infty$.

\end{lemma}

\begin{proof}

Fix $\sigma \in (0, \infty)$. It is sufficient to prove that
$P_{\theta_n, \sigma} \big ( A_n \big ) \rightarrow 1$ as $n \rightarrow \infty$. Now
\begin{equation*}
A_n = \left \{ |X_n| \le \frac{\hat{\Sigma}_n}{\sigma} \sqrt{n} \, \eta_n \right \}
\end{equation*}
where $X_n = \sqrt{n} \, \bar{Y}_n /\sigma$. Note that $X_n \sim N(\sqrt{n} \, \eta_n / 2, 1)$.
Observe that
\begin{equation*}
\left \{ \frac{\hat{\Sigma}_n}{\sigma} > \frac{3}{4} \right \} \cap
\left \{ \left | X_n - \frac{1}{2} \sqrt{n} \, \eta_n \right | \le \frac{1}{4} \sqrt{n} \, \eta_n \right \}
\subset A_n.
\end{equation*}
Thus
\begin{align*}
P_{\theta_n, \sigma} (A_n) &\ge
P_{\theta_n, \sigma} \left ( \left \{ \frac{\hat{\Sigma}_n}{\sigma} > \frac{3}{4} \right \} \cap
\left \{ \left | X_n - \frac{1}{2} \sqrt{n} \, \eta_n \right | \le \frac{1}{4} \sqrt{n} \, \eta_n \right \} \right ) \\
&= P_{\theta_n, \sigma} \left (  \frac{\hat{\Sigma}_n}{\sigma} > \frac{3}{4} \right )
P_{\theta_n, \sigma} \left (
\left | X_n - \frac{1}{2} \sqrt{n} \, \eta_n \right | \le \frac{1}{4} \sqrt{n} \, \eta_n \right )
\end{align*}
and the right-hand-side converges to 1 as $n \rightarrow \infty$.

\end{proof}
\noindent Also, when the event $B_n \cap A_n$ occurs, $\ell_n(\bar{Y}_n, \hat{\Sigma}_n) \le 0$
and $u_n(\bar{Y}_n, \hat{\Sigma}_n) \ge \theta_n$, so that
$u_n(\bar{Y}_n, \hat{\Sigma}_n) - \ell_n(\bar{Y}_n, \hat{\Sigma}_n) \ge \theta_n$. Hence,
\begin{equation*}
E_{\theta_n, \sigma} \big (\text{length of } D_n(\bar{Y}_n, \hat{\Sigma}_n) \big) \ge
P_{\theta_n, \sigma} \big ( B_n \cap A_n \big ) \, \theta_n.
\end{equation*}
Thus, for each $\sigma \in (0, \infty)$,
\begin{align*}
\frac{E_{\theta_n,\sigma} \big(\text{length of } D_n(\bar{Y}_n, \hat{\Sigma}_n) \big)}
{E_{\theta_n,\sigma} \big(\text{length of standard } 1-\alpha \text{ CI for } \theta \big)}
&\ge \frac{P_{\theta_n, \sigma} \big ( B_n \cap A_n \big ) \, \theta_n}
{2 \, t(n-1) E(\hat{\Sigma}_n)/\sqrt{n}} \\
&= \frac{P_{\theta_n, \sigma} \big ( B_n \cap A_n \big ) \, \sqrt{n} \, \eta_n}
{4 \, t(n-1) E(\hat{\Sigma}_n/\sigma)},
\end{align*}
which tends to infinity as $n \rightarrow \infty$.

\section{The motivation for considering the known error variance case}

In this appendix, we motivate the consideration of the ``known error variance case''.
{\bf We begin by supposing that the error variance $\boldsymbol{\sigma^2}$ is  unknown} and is estimated by $\hat{\sigma}^2$.
We apply the same change of
scale (by multiplying by $\sqrt{n}/\sigma$) to the parameter $\theta$, the estimator $\bar{Y}$ and the
estimators $\tilde{\Theta}_H$, $\tilde{\Theta}_S$, $\tilde{\Theta}_A$ and $\tilde{\Theta}_C$
as follows. Define $\psi = (\sqrt{n}/ \sigma) \, \theta$, $X = (\sqrt{n}/ \sigma) \, \bar{Y}$,
$\tau_n = \sqrt{n} \eta_n$ and
$R = \hat{\Sigma}/\sigma$. Note that $X$ and $R$ are independent random variables and that $X \sim N(\psi,1)$.
Also define
\begin{align*}
\tilde{\Psi}_H &= \frac{\sqrt{n}}{\sigma} \tilde{\Theta}_H  =
\begin{cases}
0 &\text{if} \ \ |X| \le R \tau_n \\
X  &\text{if} \ \ |X| > R \tau_n
\end{cases} \\
\\
\tilde{\Psi}_S &= \frac{\sqrt{n}}{\sigma} \tilde{\Theta}_S  =
\begin{cases}
-\max \{|X|-R \tau_n, 0 \} &\text{if} \ \ X < 0 \\
0  &\text{if} \ \ X=0 \\
\max \{|X|-R \tau_n, 0 \} &\text{if} \ \ X > 0
\end{cases} \\
\\
\tilde{\Psi}_A &= \frac{\sqrt{n}}{\sigma} \tilde{\Theta}_A  =
\begin{cases}
0 &\text{if} \ \ |X| \le R \tau_n \\
X - \displaystyle{\frac{R^2 \tau_n^2}{X}} &\text{if} \ \ |X| > R \tau_n
\end{cases} \\
\\
\tilde{\Psi}_C &= \frac{\sqrt{n}}{\sigma} \tilde{\Theta}_C  =
\begin{cases}
\text{sign}(X) (|X| - R \tau_n)_+ &\text{if} \ \ |X| \le 2 R \tau_n \\
\big((a-1) X -\text{sign}(X) \, a \, R \, \tau_n \big)/(a-2)  &\text{if} \ \ 2 R \tau_n < |X| \le a R \tau_n \\
X &\text{if} \ \ |X| > a R \tau_n
\end{cases}
\end{align*}
These are not estimators of $\psi$ since they depend on the unknown parameter $\sigma$.
Since $R$ and $X$ are independent and $R$ converges in probability to 1 (as $n \rightarrow \infty$)
it is plausible that, for large $n$, the statistical properties of $\tilde{\Psi}_H$, $\tilde{\Psi}_S$, $\tilde{\Psi}_A$
and $\tilde{\Psi}_C$ are well-approximated by these properties of the corresponding quantities:
\begin{align*}
\hat{\Psi}_H &=
\begin{cases}
0 &\text{if} \ \ |X| \le \tau_n \\
X  &\text{if} \ \ |X| > \tau_n
\end{cases} \\
\\
\hat{\Psi}_S &=
\begin{cases}
-\max \{|X|- \tau_n, 0 \} &\text{if} \ \ X < 0 \\
0  &\text{if} \ \ X=0 \\
\max \{|X|- \tau_n, 0 \} &\text{if} \ \ X > 0
\end{cases} \\
\\
\hat{\Psi}_A &=
\begin{cases}
0 &\text{if} \ \ |X| \le \tau_n \\
X - \displaystyle{\frac{\tau_n^2}{X}} &\text{if} \ \ |X| > \tau_n
\end{cases} \\
\\
\hat{\Psi}_C &=
\begin{cases}
\text{sign}(X) (|X| -  \tau_n)_+ &\text{if} \ \ |X| \le 2  \tau_n \\
\big((a-1) X -\text{sign}(X) a \tau_n \big)/(a-2)  &\text{if} \ \ 2 \tau_n < |X| \le a \tau_n \\
X &\text{if} \ \ |X| > a \tau_n
\end{cases}
\end{align*}
Note that, conveniently, the statistical properties of these quantities depend only on the parameter $\psi$ and
not on the parameter $\sigma$.

Farchione and Kabaila \cite{Farchione_Kabaila} consider the following confidence interval for $\theta$:
\begin{equation*}
D = \left [ -\frac{\hat{\Sigma}}{\sqrt{n}} \, b \left ( -\frac{\bar{Y}}{\hat{\Sigma}/\sqrt{n}} \right ), \,
\frac{\hat{\Sigma}}{\sqrt{n}} \, b \left ( \frac{\bar{Y}}{\hat{\Sigma}/\sqrt{n}} \right ) \right ]
\end{equation*}
where the function $b$ must satisfy the constraint that $b(x) \ge -b(-x)$ for all $x \in \mathbb{R}$.
This constraint is required to ensure that the upper endpoint of this confidence interval is never less
than the lower endpoint. This particular form of confidence interval is motivated by some invariance arguments.
The standard $1-\alpha$ confidence interval for $\theta$ is
\begin{equation*}
\big [ \bar{Y} - t(n-1) \hat{\Sigma} / \sqrt{n}, \, \bar{Y} + t(n-1) \hat{\Sigma} / \sqrt{n} \big ]
\end{equation*}
where the quantile $t(m)$ is defined by the requirement that $P(-t(m) \le T \le t(m)) = 1 - \alpha$
for $T \sim t_m$. Note that this confidence interval can be expressed in the form $D$.

Now scale the confidence interval $D$ by the same scaling factor as before, to obtain
\begin{equation*}
D^* = \frac{\sqrt{n}}{\sigma} D = \big [ -R \, b(-X/R), \, R \, b(X/R) \big ].
\end{equation*}
Note that $\tilde{\Theta}_H \in D$ is equivalent to $\tilde{\Psi}_H \in D^*$.
Similar statements apply to the other estimators $\tilde{\Theta}_S$,
$\tilde{\Theta}_A$ and $\tilde{\Theta}_C$. Also note that
$D^*$ is not a confidence interval for $\psi$, since it depends on the unknown parameter $\sigma$. However,
\begin{equation*}
P_{\theta, \sigma} (\theta \in D) = P_{\psi} (\psi \in D^*),
\end{equation*}
so that
\begin{equation*}
\inf_{\theta, \sigma} P_{\theta, \sigma} (\theta \in D) = \inf_{\psi} P_{\psi} (\psi \in D^*).
\end{equation*}
Also,
\begin{equation*}
\frac{E_{\theta, \sigma} (\text{length of } D)}{E_{\theta, \sigma} (\text{length of standard } 1-\alpha \text{ CI for } \theta)}
= \frac{E_{\psi} ( \text{length of } D^*)}{2 t(n-1) E(R)}.
\end{equation*}

Since $R$ and $X$ are independent and $R$ converges in probability to 1 (as $n \rightarrow \infty$)
it is plausible that, for large $n$, the statistical properties of $D^*$ are well-approximated by
the corresponding properties of $C^* = [-b(-X), b(X)]$. In fact, the following result holds.

\begin{theorem}

Suppose that the function $b$ satisfies the following assumptions.

\begin{enumerate}

\item[(A1)] The function $b$ is continuous and strictly increasing. Also, the function $b^{-1}$ is uniformly
continuous.

\item[(A2)] Define $e(x) = b(x) - x - z$, where the quantile $z$ is defined by the requirement that $P(-z \le Z \le z) = 1 - \alpha$
for $Z \sim N(0,1)$.

\begin{enumerate}

\item[(i)] $e(x) = 0$ for all $|x| \ge q$, where $q$ is a specified positive number.

\item[(ii)] There exists $L$, satisfying $0 < L < \infty$, such that $|e(x) - e(y)| \le L |x-y|$ for all $x$ and $y$.

\end{enumerate}

\end{enumerate}

Then

\begin{enumerate}

\item[(R1)] $\displaystyle{\sup_{\psi} \big|P_{\psi} (\psi \in C^*) -  P_{\psi} (\psi \in D^*) \big| \rightarrow 0}$
as $n \rightarrow \infty$.

\item[(R2)] $\displaystyle{\sup_{\psi} \left |
\frac{E_{\psi} (\text{length of } C^*)}{2 z} - \frac{E_{\psi} (\text{length of } D^*)}{2 t(n-1) E(R)}  \right |
\rightarrow 0}$ as $n \rightarrow \infty$.

\end{enumerate}

\end{theorem}

\begin{proof}

We prove the result {\it (R1)} as follows. Note that
\begin{align*}
P_{\psi} (\psi \in C^*) &= 1 - P_{\psi} (\psi < - b(-X)) - P_{\psi} (\psi > b(X)) \\
P_{\psi} (\psi \in D^*) &= 1 - P_{\psi} (\psi < - R \, b(-X/R)) - P_{\psi} (\psi > R \, b(X/R)).
\end{align*}
It is sufficient to prove that
\begin{align}
\label{first_limit}
&\sup_{\psi} \big | P_{\psi} \big(\psi < - b(-X) \big) - P_{\psi} \big(\psi < - R \, b(-X/R) \big) \big | \rightarrow 0 \quad \text{as } n \rightarrow \infty \\
\label{second_limit}
&\sup_{\psi} \big | P_{\psi} \big(\psi > b(X) \big) - P_{\psi} \big(\psi > R \, b(X/R) \big) \big | \rightarrow 0 \quad \text{as } n \rightarrow \infty
\end{align}
The proofs of \eqref{first_limit} and \eqref{second_limit} are very similar. For the sake of brevity, we provide
only the proof of \eqref{second_limit}.
Suppose that $\epsilon > 0$ is given. We need to prove that there exists $N < \infty$ such that
\begin{equation}
\label{result_to_prove}
\sup_{\psi} \big | P_{\psi} (\psi > b(X)) - P_{\psi} (\psi > R \, b(X/R)) \big | < \epsilon \quad \text{for all } n > N.
\end{equation}

Let $\delta$ ($0<\delta<1/2$) be given. Using the law of total probability, it may be shown that
\begin{align}
\notag
&\big | P_{\psi} \big(\psi > R \, b(X/R) \big) - P_{\psi} \big(\psi > b(X) \big) \big | \\
\label{initial}
&\le \big | P_{\psi} \big(\psi > R \, b(X/R), \, |R-1| \le \delta \big) - P_{\psi} \big(\psi > b(X) \big) \big | + P(|R-1| > \delta).
\end{align}
Obviously,
\begin{align}
\notag
&P_{\psi} \big(\psi > R \, b(X/R), \, |R-1| \le \delta \big)  \\
\label{endgame}
&= P_{\psi} \big( \psi > b(X) + (R-1)(z + e(X/R)) + (e(X/R)-e(X)), |R-1| \le \delta \big).
\end{align}
It may be shown that if $|R-1| \le \delta$ then there exists $M < \infty$ (where $M$ does not depend on $\delta$)
such that $\big|(R-1)(z + e(X/R))+ (e(X/R)-e(X)) \big| \le M \delta$. Thus
\begin{equation*}
P_{\psi} \big( \psi > b(X) + M \delta, |R-1| \le \delta \big) \le \eqref{endgame} \le P_{\psi} \big( \psi > b(X) - M \delta \big).
\end{equation*}
Using the law of total probability, it may be shown that
\begin{equation*}
P_{\psi} \big( \psi > b(X) + M \delta, |R-1| \le \delta \big) \le \eqref{endgame} \ge P_{\psi} \big( \psi > b(X) + M \delta \big)
- P(|R-1| > \delta).
\end{equation*}
Thus
\begin{equation*}
P_{\psi} \big( \psi > b(X) + M \delta \big) - P(|R-1| > \delta) \le \eqref{endgame} \le P_{\psi} \big( \psi > b(X) - M \delta \big).
\end{equation*}
In other words,
\begin{equation*}
P_{\psi} \big( X < b^{-1}(\psi - M \delta)\big) - P(|R-1| > \delta)
\le \eqref{endgame} \le P_{\psi} \big( X < b^{-1}(\psi + M \delta)\big).
\end{equation*}
Note that $P_{\psi} (\psi > b(X)) = P_{\psi} (X < b^{-1}(\psi))$.
Using the uniform continuity of $b^{-1}$ and the fact that $X \sim N(\psi,1)$, it may be shown
that there exists $\delta$ ($0<\delta<1/2$) such that
\begin{align*}
&\sup_{\psi} \big |P_{\psi} \big( X < b^{-1}(\psi - M \delta) \big) - P_{\psi} \big( X < b^{-1}(\psi) \big) \big | < \epsilon/2 \\
&\sup_{\psi} \big |P_{\psi} \big( X < b^{-1}(\psi + M \delta) \big) - P_{\psi} \big( X < b^{-1}(\psi) \big) \big | < \epsilon/2.
\end{align*}
Choose $\delta$ ($0<\delta<1/2$) such that these two inequalities are satisfied.
Therefore, $\big|\eqref{endgame} - P_{\psi} ( \psi > b(X)) \big| < P(|R-1| > \delta) + \epsilon/2$.
It follows from \eqref{initial} that
$\big | P_{\psi} \big(\psi > R \, b(X/R) \big) - P_{\psi} (\psi > b(X)) \big | < 2 P(|R-1| > \delta) + \epsilon/2$.
Since $P(|R-1| > \delta) \rightarrow 0$ as $n \rightarrow \infty$, there exists $N < \infty$ such that
\eqref{result_to_prove} is satisfied. This completes the proof of the result {\it (R1)}.

We prove the result {\it (R2)} as follows. It may be shown that it is sufficient to prove that
\begin{equation*}
\sup_{\psi} \big | E_{\psi} (\text{length of } C^*) - E_{\psi} (\text{length of } D^*) \big |
\rightarrow 0 \quad \text{as } n \rightarrow \infty.
\end{equation*}
Now
\begin{align*}
&E_{\psi} (\text{length of } D^*) - E_{\psi} (\text{length of } C^*) \\
&= 2 z (E(R)-1) + \big ( E_{\psi} (\text{length of } D^*) - 2 z E(R) \big )
- \big ( E_{\psi} (\text{length of } C^*) - 2 z \big ).
\end{align*}
Hence
\begin{align*}
&\big | E_{\psi} (\text{length of } D^*) - E_{\psi} (\text{length of } C^*) \big | \\
&= 2 z |E(R)-1| + \big | E_{\psi} (\text{length of } D^*) - 2 z E(R) \big )
- \big ( E_{\psi} (\text{length of } C^*) - 2 z \big ) \big |.
\end{align*}
Since $E(R)$ does not depend on $\psi$ and $E(R) \rightarrow 1$ as $n \rightarrow \infty$, it is
sufficient to prove that
\begin{equation*}
\sup_{\psi} \big | E_{\psi} (\text{length of } D^*) - 2 z E(R) \big )
- \big ( E_{\psi} (\text{length of } C^*) - 2 z \big ) \big |
\rightarrow 0 \quad \text{as } n \rightarrow \infty.
\end{equation*}
Let $f_R$ denote the probability density function of $R$. Now
\begin{align}
\notag
&E_{\psi} (\text{length of } D^*) - 2 z E(R) \\
\notag
&= \int_0^{\infty} \int_{-\infty}^{\infty}
\left ( b \left ( \frac{x}{r} \right ) + b \left ( -\frac{x}{r} \right ) -2 z \right )
\phi(x - \psi) \, dx \, r \, f_R(r) \, dr \\
\label{intermediate_d_star}
&= \int_0^{\infty} \int_{-rq}^{rq}
\left ( b \left ( \frac{x}{r} \right ) + b \left ( -\frac{x}{r} \right ) -2 z \right )
\phi(x - \psi) \, dx \, r \, f_R(r) \, dr
\end{align}
since $b(x/r) + b(-x/r) - 2 z = 0$ for all $|x| \ge rq$. Changing the variable of integration
from $x$ to $y=x/r$, we see that \eqref{intermediate_d_star} is equal to
\begin{equation*}
\int_0^{\infty} \int_{-q}^q
\left ( b \left ( x \right ) + b \left ( -x \right ) -2 z \right )
\phi(r x - \psi) \, dx \, r^2 \, f_R(r) \, dr
\end{equation*}

Now
\begin{align}
\notag
&E_{\psi} (\text{length of } C^*) - 2 z \\
\notag
&= \int_{-\infty}^{\infty}
\left ( b \left ( x \right ) + b \left ( -x \right ) -2 z \right )
\phi(x - \psi) \, dx \\
\notag
&= \int_{-q}^q
\left ( b \left ( x \right ) + b \left ( -x \right ) -2 z \right )
\phi(x - \psi) \, dx  \ \ \ \ (\text{by } {\it (A2) (i)}) \\
\label{intermediate_c_star}
&= \int_0^{\infty} \int_{-q}^q
\left ( b \left ( x \right ) + b \left ( -x \right ) -2 z \right )
\phi(x - \psi) \, dx \, r^2 \, f_R(r) \, dr,
\end{align}
since
\begin{equation*}
\int_0^{\infty} r^2 f_R(r) dr = E(R^2) = 1.
\end{equation*}
Thus
\begin{align}
\notag
&E_{\psi} (\text{length of } D^*) - 2 z E(R)
- \big ( E_{\psi} (\text{length of } C^*) - 2 z \big ) \\
\label{intermediate_c_star}
&\int_0^{\infty} \int_{-q}^q (e(x) + e(-x)) \big ( \phi(r x - \psi) - \phi(x - \psi) \big )
\, dx \, r^2 \, f_R(r) \, dr
\end{align}
By the mean-value theorem, there exists a positive number $K < \infty$ such that
\begin{equation*}
\big | \phi(r x - \psi) - \phi(x - \psi) \big | \le K |r-1| |x| \quad
\text{for all } r \ge 0 \ \text{and } x \in \mathbb{R}.
\end{equation*}
Thus
\begin{equation*}
| \, \eqref{intermediate_c_star} \, | \le 4 L K q^3 E \big(|R-1| R^2 \big).
\end{equation*}
Note that $E \big(|R-1| R^2 \big)$ does not depend on $\psi$ and that, by the Cauchy-Schwarz inequality,
$E \big(|R-1| R^2 \big) \rightarrow 0$ as $n \rightarrow \infty$. This completes the proof of {\it (R2)}.

\end{proof}

\noindent \textbf{Thus, to study the coverage and expected length properties of the confidence interval
$\boldsymbol{D}$ for $\boldsymbol{\theta}$ when $\boldsymbol{n}$ is large,
we study the properties of $\boldsymbol{P_{\psi}(\psi \in C^*)}$ and $\boldsymbol{E_{\psi} (\text{length of } C^*)}$,
which are simply functions of $\boldsymbol{\psi}$.}

Now suppose that the ``error variance is known'' i.e. $\sigma^2$ is known. The analogues of the estimators
$\tilde{\Theta}_H$, $\tilde{\Theta}_S$, $\tilde{\Theta}_A$ and $\tilde{\Theta}_C$ are
$\hat{\Theta}_H$, $\hat{\Theta}_S$, $\hat{\Theta}_A$ and $\hat{\Theta}_C$ respectively, where
\begin{align*}
\hat{\Theta}_H &=
\begin{cases}
0 &\text{if} \ \ |\bar{Y}| \le \sigma \, \eta_n \\
\bar{Y} &\text{if} \ \ |\bar{Y}| > \sigma \, \eta_n
\end{cases} \\
\hat{\Theta}_S &=
\begin{cases}
- \max\{|\bar{Y}| - \sigma \eta_n, 0\} &\text{if } \bar{Y}<0 \\
0 &\text{if } \bar{Y}=0 \\
\max\{|\bar{Y}| - \sigma \eta_n, 0\} &\text{if } \bar{Y}>0
\end{cases} \\
\hat{\Theta}_A &=
\begin{cases}
0 &\text{if } |\bar{Y}| \le \sigma \eta_n \\
\bar{Y} - \displaystyle{\frac{\sigma^2 \eta_n^2}{\bar{Y}}} &\text{if } |\bar{Y}| > \sigma \eta_n
\end{cases} \\
\hat{\Theta}_C &=
\begin{cases}
\text{sign}(\bar{Y}) \, \big(|\bar{Y}| - \sigma \eta_n \big)_+
&\text{if } |\bar{Y}| \le 2 \sigma \eta_n \\
((a-1) \bar{Y} - \text{sign}(\bar{Y}) a \sigma \eta_n)/(a-2)
&\text{if }  2 \sigma \eta_n < |\bar{Y}| \le a \sigma \eta_n \\
\bar{Y}
&\text{if }  |\bar{Y}| > a \sigma \eta_n
\end{cases}
\end{align*}
where $a=3.7$. Also, the analogue of the confidence interval $D$ for $\theta$
is
\begin{equation*}
C = \left [ -\frac{\sigma}{\sqrt{n}} b \left ( -\frac{\bar{Y}}{\sigma/\sqrt{n}} \right ), \,
\frac{\sigma}{\sqrt{n}} b \left ( \frac{\bar{Y}}{\sigma/\sqrt{n}} \right ) \right ].
\end{equation*}
Scaling $\theta$, $\hat{\Theta}_H$, $\hat{\Theta}_S$, $\hat{\Theta}_A$, $\hat{\Theta}_C$ and $C$ by
multiplying by $\sqrt{n}/\sigma$, we obtain $\psi$, $\hat{\Psi}_H$, $\hat{\Psi}_S$, $\hat{\Psi}_A$, $\hat{\Psi}_C$ and $C^*$,
respectively. In other words, when we suppose that the ``error variance is known'', we are finding an approximation
(by the arguments stated earlier in this section) to the coverage probability and expected length properties of
$\tilde{\Theta}_H$, $\tilde{\Theta}_S$, $\tilde{\Theta}_A$, $\tilde{\Theta}_C$ and $D$
for large $n$.

\section*{Acknowledgements}
Paul Kabaila is grateful to Hannes Leeb and Benedikt P\"otscher for
some helpful discussions.


\begin{thebibliography}{9}






\bibitem{Fan_Li}
\textsc{Fan, J.} and
\textsc{Li, R.}(2001).
Variable selection via nonconcave penalized likelihood and
its oracle properties. \textit{Journal of the American Statistical Association}
\textbf{96} 1348--1360.

\bibitem{Farchione_Kabaila}
\textsc{Farchione, D.} and
\textsc{Kabaila, P.}(2008).
Confidence intervals for the normal mean utilizing prior information. \textit{Statistics \& Probability Letters}
\textbf{78} 1094--1100.

\bibitem{Kabaila}
\textsc{Kabaila, P.}(1995).
The effect of model selection on confidence regions and prediction regions.
\textit{Econometric Theory}
\textbf{11} 537--549.


\bibitem{Potscher}
\textsc{P\"otscher, B.}(2009).
Confidence sets based on sparse estimators are necessarily
large. \textit{Sankhya}
\textbf{71-A} 1--18.

\bibitem{Potscher_Schneider}
\textsc{P\"otscher, B.} and
\textsc{Schneider, U.}(2010).
Confidence sets based on penalized maximum likelihood estimators in
Gaussian regression. \textit{Electronic Journal of Statistics}
\textbf{4} 334--360.

\bibitem{Pratt}
\textsc{Pratt, J.W.}(1961).
Length of confidence intervals. \textit{Journal of the American Statistical Association}
\textbf{56} 549--657.

\bibitem{Tibshirani}
\textsc{Tibshirani, R.}(2010).
Regression shrinkage and selection via the lasso.
\textit{Journal of the Royal Statistical Society, Series B}
\textbf{58} 267--288.

\bibitem{Zou}
\textsc{Zou, H.}(2006).
The adaptive lasso and its oracle properties.
\textit{Journal of the American Statistical Association}
\textbf{101} 1418--1429.


\end{thebibliography}
\end{document}